\documentclass[10pt, conference, compsocconf]{IEEEtran}
\IEEEoverridecommandlockouts
\usepackage{cite}
\ifCLASSINFOpdf

\else
 
\fi

\usepackage[utf8]{inputenc} 
\usepackage[T1]{fontenc}    
\usepackage{hyperref}       
\usepackage{url}            
\usepackage{booktabs}       
\usepackage{amsfonts}       
\usepackage{nicefrac}       
\usepackage{microtype}      
\let\labelindent\relax
\usepackage{enumitem}

\usepackage{amsmath}

\usepackage{graphicx}
\usepackage{subcaption}
\usepackage{colortbl}
\definecolor{LightCyan}{rgb}{0.88,1,1}

\usepackage{color}
\usepackage[linesnumbered,ruled,vlined ]{algorithm2e}
\SetKwComment{Comment}{ }{}
\usepackage{wrapfig}

\newcommand{\A}{\mathbf{A}} 
\newcommand{\ba}{\mathbf{a}} 
\newcommand{\tA}{\tilde{\mathbf{A}}} 
\newcommand{\x}{\mathbf{x}} 
\newcommand{\bb}{\mathbf{b}} 
\newcommand{\B}{\mathbf{B}} 
\newcommand{\tB}{\tilde{\mathbf{B}}} 
\newcommand{\y}{\mathbf{y}}
\newcommand{\C}{\mathbf{C}}
\newcommand{\cc}{\mathbf{c}}
\newcommand{\s}{\mathbf{S}}

\newcommand{\R}{\mathbb{R}}
\newcommand{\e}{\mathbb{E}}
\newcommand{\nn}{\nonumber}

\allowdisplaybreaks

\usepackage{amsthm}

\newtheorem{theorem}{Theorem}[section]

\newtheorem{lemma}[theorem]{Lemma}
\theoremstyle{definition}
\newtheorem{definition}{Definition}[section]
\newtheorem{remark}{Remark}

\renewcommand{\thesection}{\arabic{section}} 
\renewcommand{\thesubsection}{\arabic{subsection}}

\def\IEEEbibitemsep{2pt plus .5pt}

\begin{document}
%
\title{OverSketch: Approximate Matrix Multiplication for the Cloud}




%
\author{\IEEEauthorblockN{Vipul Gupta$^\star$,
Shusen Wang$^\dag$, Thomas Courtade$^\star$ and Kannan Ramchandran$^\star$\thanks{This work was supported by NSF Grants CCF-1703678, CCF-1704967 and CCF-0939370 (Center for Science of Information)}}
\IEEEauthorblockA{$^\star$Department of EECS, UC Berkeley\\ 
$^\dag$Department of CS, Stevens Institute of Technology\\
Email: \{vipul\_gupta, courtade, kannanr\}@eecs.berkeley.edu, shusen.wang@stevens.edu
}}

\maketitle

\begin{abstract}

We propose OverSketch, an approximate algorithm for distributed matrix multiplication in serverless computing. 
OverSketch leverages ideas from matrix sketching and high-performance computing to enable cost-efficient multiplication that is resilient to faults and straggling nodes pervasive in low-cost serverless architectures. We establish statistical guarantees on the accuracy of OverSketch and empirically validate our results by solving a large-scale linear program using interior-point methods and demonstrate a $34\%$ reduction in compute time on AWS Lambda.
\end{abstract}

\begin{IEEEkeywords}
serverless computing, straggler mitigation, sketched matrix multiplication
\end{IEEEkeywords}

%
\IEEEpeerreviewmaketitle

\section{Introduction}

Matrix multiplication is a frequent computational bottleneck in fields like scientific computing, machine learning, graph processing, etc. In many applications, such matrices are very large, with dimensions easily scaling up to millions. 
 Consequently, the last three decades have witnessed the development of many algorithms for parallel matrix multiplication for High Performance Computing (HPC). During the same period, technological trends like Moore's law made arithmetic operations faster and, as a result, the bottleneck for parallel computation shifted from computation to communication. Today, the cost of moving data between nodes exceeds the cost of arithmetic operations by orders of magnitude, and this gap is increasing exponentially with time \cite{book_future_supercomputing, comm_cost2, demmel}. This has led to the popularity of communication-avoiding algorithms for parallel computation \cite{demmel, 2.5d}.

In the last few years, there has been a paradigm shift from HPC towards distributed computing on the cloud due to extensive and inexpensive commercial offerings. In spite of developments in recent years, server-based cloud computing is inaccessible to a large number of users due to complex cluster management and a myriad of configuration tools. Serverless computing\footnote{The  term `serverless' is a misnomer, servers are still used for computation but their maintenance and provisioning is hidden from the user.} has recently begun to fill this void by abstracting away the need for maintaining servers and thus removing the need for complicated cluster management while providing greater elasticity and easy scalability \cite{serverless_computing, pywren, numpywren}. Some examples are Amazon Web Services (AWS) based Lambda, Microsoft Azure functions, and Google Cloud Functions.
Large-scale matrix multiplication, being embarrassingly parallel and frequently encountered, is a natural fit for serverless computing. 

\begin{figure}[t]
        \centering
        \includegraphics[scale=0.33]{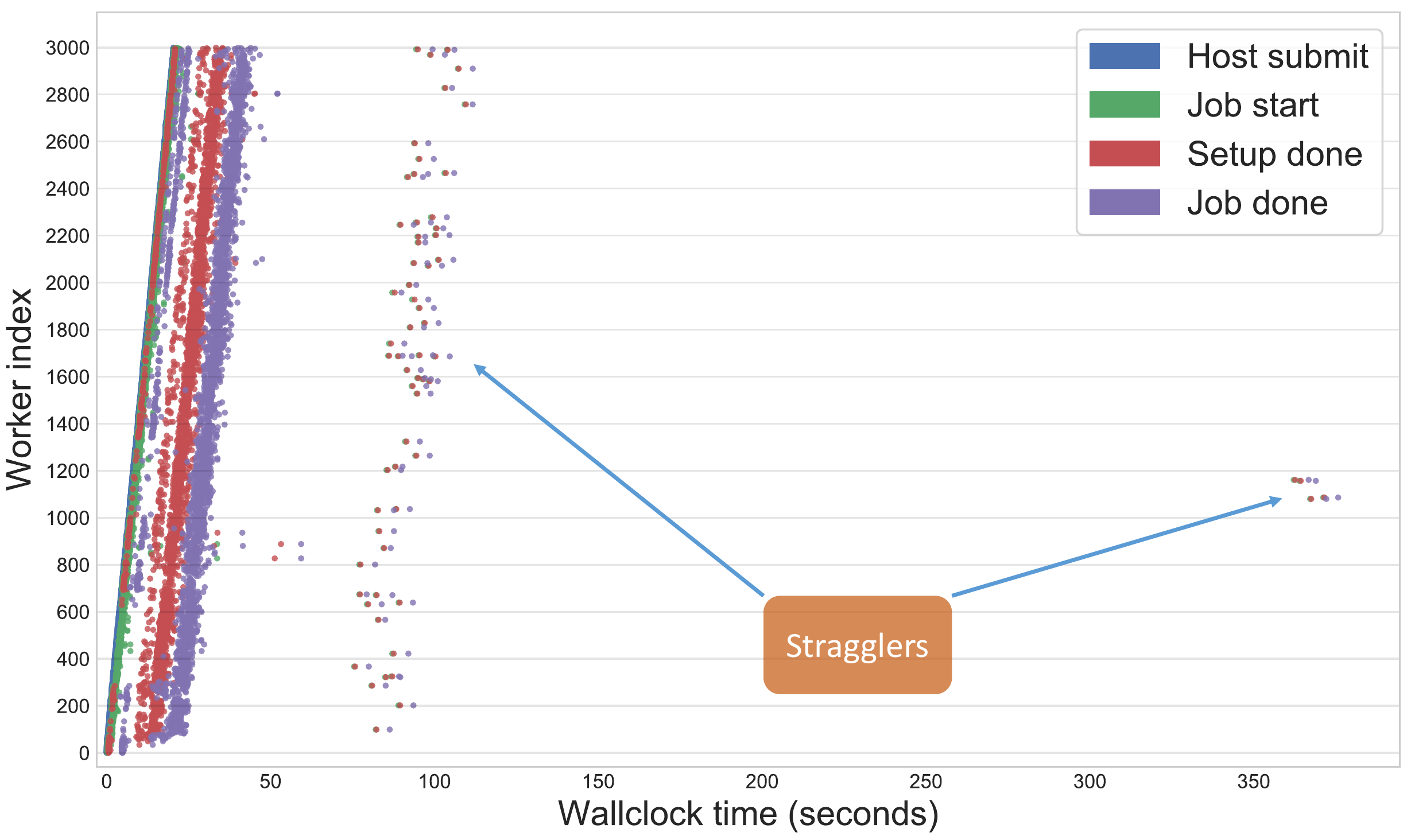}
    \caption{\small Job times for 3000 AWS Lambda nodes where the median job time is around 40 seconds, and around 5\% of the nodes take 100 seconds, and two nodes take as much as 375 seconds to complete the same job.}
    \label{fig:stragglers}
\vspace{-3mm}
\end{figure}

Existing distributed algorithms for HPC/server-based systems cannot, in general, be extended to serverless computing due to the following crucial differences between the two architectures:
\begin{itemize}[leftmargin=*]
\setlength{\itemindent}{0em}
\item Workers in the serverless setting, unlike cluster nodes, do not communicate amongst themselves. They read/write data directly from/to a single data storage entity (for example, cloud storage like AWS S3) and the user is only allowed to submit prespecified jobs and does not have any control over the management of workers \cite{serverless_computing, pywren, numpywren}.
\item Distributed computation in HPC/server-based systems is generally limited by the number of workers at disposal. However, in serverless systems, the number of inexpensive workers can easily be scaled into the thousands, but these low-commodity nodes are generally limited by the amount of memory and lifespan available.
\item Unlike HPC,  nodes in the cloud-based systems suffer degradation due to system noise which can be a result of limited availability of shared resources, network latency, hardware failure, etc. \cite{tailatscale, hoefler}. This causes variability in job times, which results in subsets of slower nodes, often called \emph{stragglers}, which significantly slow the computation. Time statistics for worker job times are plotted in Figure \ref{fig:stragglers} for AWS Lambda. Notably, there are a few workers ($\sim5\%$) that take much longer than the median job time, thus decreasing the overall computational efficiency of the system. Distributed algorithms robust to such unreliable nodes are desirable in cloud computing. 
\end{itemize}


\subsection{Main Contributions}

This paper bridges the gap between communication-efficient algorithms for distributed computation and existing methods for straggler-resiliency. To this end, we first analyze the monetary cost of distributed matrix multiplication for serverless computing for two different schemes of partitioning and distributing the data. Specifically, we show that row-column partitioning of input matrices requires asymptotically more communication than blocked partitioning for distributed matrix multiplication, similar to the optimal communication-avoiding algorithms in the HPC literature. 

In applications like machine learning, where the data itself is noisy, solution accuracy is often traded for computational efficiency. Motivated by this, we propose OverSketch, a sketching scheme to perform blocked \emph{approximate} matrix multiplication and prove statistical guarantees on the accuracy of the result. 
OverSketch has threefold advantages:
\begin{enumerate}[leftmargin=*]
\item Reduced computational complexity by significantly decreasing the dimension of input matrices using sketching,
\item Resiliency against stragglers and faults in serverless computing by over-provisioning the sketch dimension,
\item Communication efficiency for distributed multiplication due to the blocked partition of input matrices. 
\end{enumerate}
Sketching for OverSketch requires linear time that is embarrassingly parallel. 
Through experiments on AWS Lambda, we show that small redundancy ($\approx 5\%$) is enough to tackle stragglers using OverSketch. Furthermore, we use OverSketch to calculate the Hessian distributedly while solving a large linear program using interior point methods and demonstrate a $34\%$ reduction in total compute time on AWS Lambda. 


\subsection{Related Work}
\label{related_work}

Traditionally, techniques like speculative execution are used to deal with stragglers, for example, Hadoop MapReduce \cite{mapreduce} and Apache Spark \cite{spark}. Such techniques work by detecting nodes that are running slowly or will slow down in the future and then assigning their jobs to new nodes without shutting down the original job. The node that finishes first submits its results. This has many limitations. A constant monitoring of jobs is required, which might be costly if there are many workers in the system. It is also possible that a node will straggle only towards the end of the job, and by the time the job is resubmitted, the additional time and computational overhead has already hurt the overall efficiency of the system. The situation is even worse for smaller jobs, as spinning up an extra node requires additional invocation and setup time which can exceed the job time itself.

Recently, approaches based on coding theory have been developed which cleverly introduce redundancy into the computation to deal with stragglers \cite{kangwook1,kangwook2,tavor,poly_codes,matdot,jingge}. Many of these proposed schemes have been dedicated to distributed matrix multiplication \cite{kangwook2, poly_codes, tavor, matdot}. In \cite{kangwook2}, the authors develop a coding scheme for matrix multiplication that uses Maximum Distance Separable (MDS) codes to code $\A$ in a column-wise fashion and $\B$ in a row-wise fashion, so that the resultant is a product-code of $\C$, where $\C = \A\times \B$. An illustration is shown in Figure \ref{fig:kangwook2}. A simpler version of this has been known in the HPC community as Algorithm-Based-Fault-Tolerance (ABFT) \cite{abft}. 
 Authors  in \cite{tavor} generalize the results in \cite{kangwook2} to a $d$-dimensional product code with only one parity in each dimension. In \cite{poly_codes}, the authors develop polynomial codes for matrix multiplication, which is an improvement over \cite{kangwook2} in terms of recovery threshold, that is, the minimum number of workers required to recover the product $\C$. 
 
The commonality in these and other similar results is that they divide the input matrices into row and column blocks, where each worker multiplies a row block (or some combination of row blocks) of $\A$ and a column block (or some combination of column blocks) of $\B$. These methods provide straggler resiliency but are not cost-efficient as they require asymptotically more communication than blocked partitioning of data, as discussed in detail in the next section. Another disadvantage of such coding-based methods is that there are separate encoding and decoding phases that require additional communication and potentially large computational burden at the master node, which may make the algorithm infeasible in some distributed computing environments. 


\begin{figure}[t]
\includegraphics[width=0.48\textwidth]{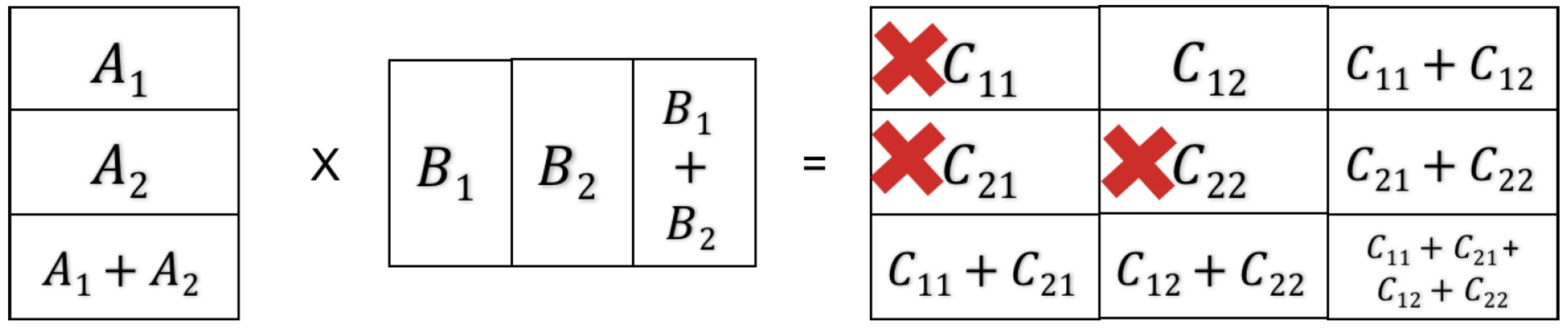}
\centering
\caption{ \small
Matrix $\A$ is divided into 2 row chunks $\A_1$ and $\A_2$, while $\B$ is divided into two column chunks $\B_1$ and $\B_2$. 
During the encoding process, redundant chunks $\A_1+\A_2$ and $\B_1+\B_2$ are created. To compute $\C$, 9 workers store each possible combination of a chunk of $\A$ and $\B$ and multiply them. During the decoding phase, the master can recover the affected data ($C_{11}, C_{12}$ and $C_{22}$ in this case) using the redundant chunks.} 
\label{fig:kangwook2}
\vspace{-3mm}
\end{figure}

\section{Preliminaries}
\label{sec:motivation}

\begin{figure}[t]
    \centering
    \begin{subfigure}[t]{0.5\textwidth}
        \centering
        \includegraphics[scale=0.45]{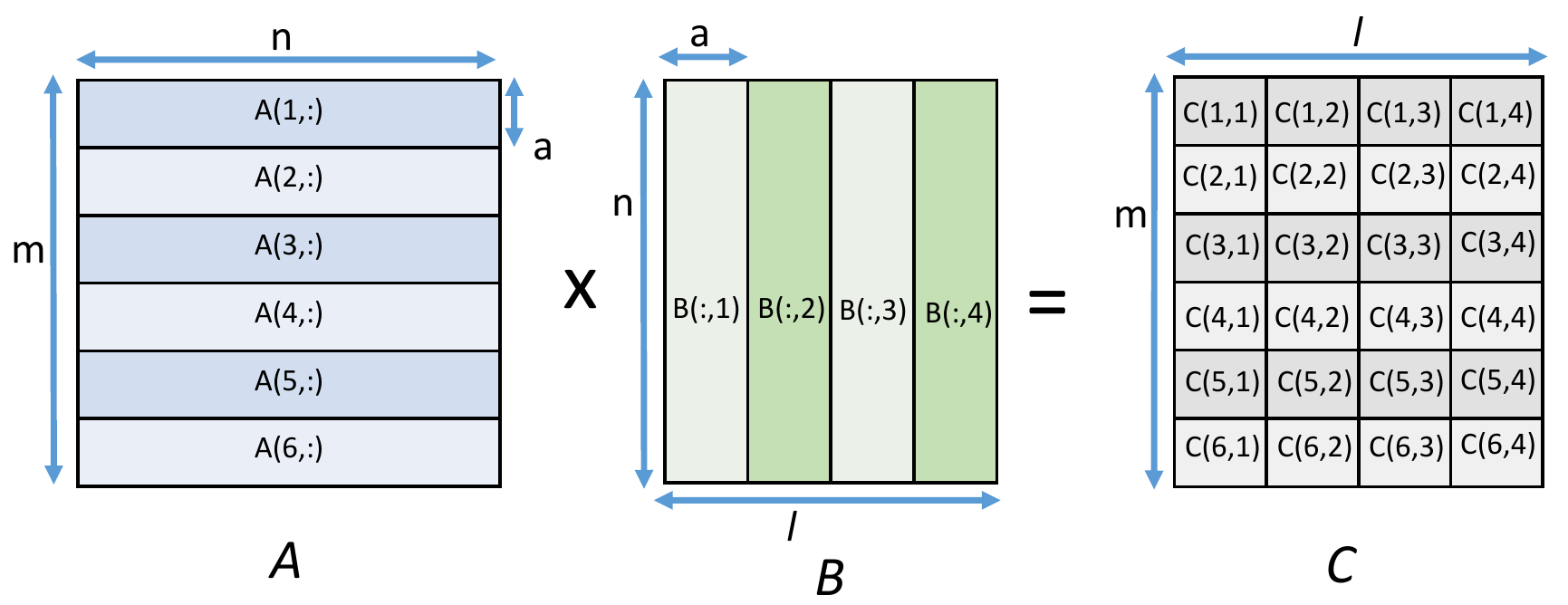}
        \caption{\footnotesize Distributed naive matrix multiplication, where each worker multiplies a row-block of $\A$ of size $a\times n$ and a column block of $\B$ of size $n\times a$ to get an $a\times a$ block of $\C$.}
        \label{fig:gemm_naive}
    \end{subfigure}%
    \hspace{0.02\textwidth}
    \begin{subfigure}[t]{0.5\textwidth}
        \centering
        \includegraphics[scale=0.525]{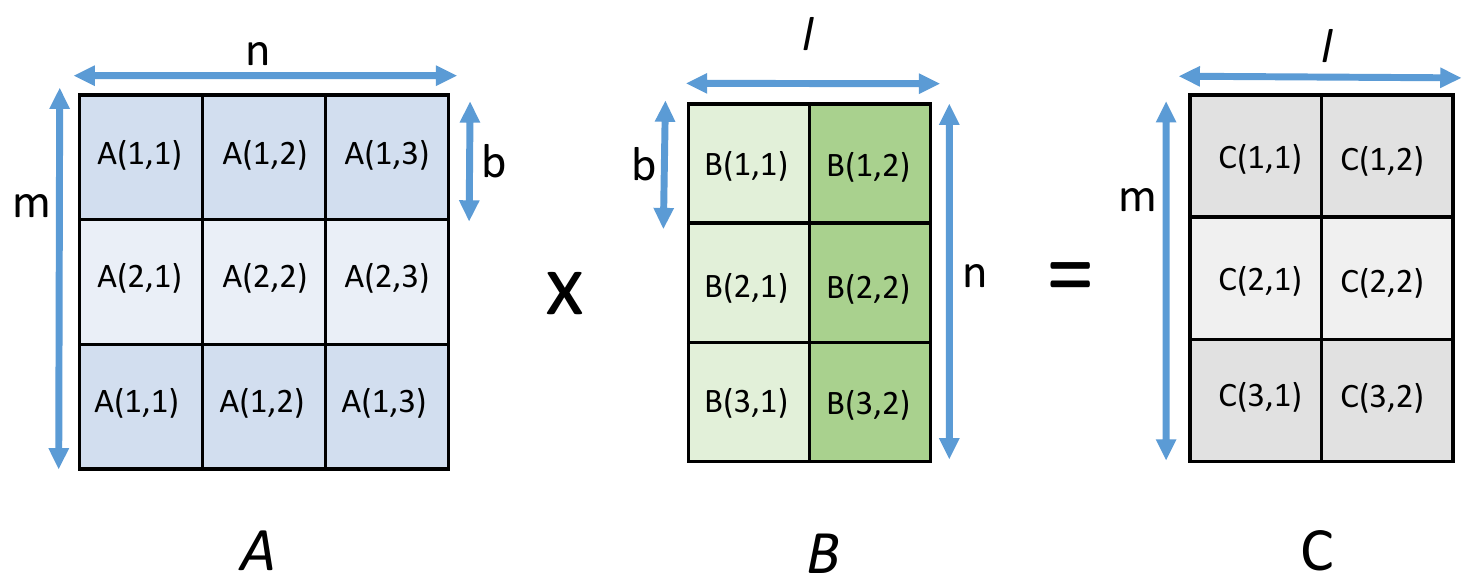}
        \caption{\footnotesize Distributed blocked matrix multiplication, where each worker multiplies a sub-block of $\A$ of size $b\times b$ and a sub-block of $\B$ of size $b\times b$.}
        \label{fig:gemm_blocked}
    \end{subfigure}
    \caption{\small An illustration of two algorithms for distributed matrix multiplication.}
    \label{fig:gemm}
\vspace{-3mm}
\end{figure}

 There are two common schemes for distributed multiplication of two matrices $\A\in \R^{m\times n}$ and $\B\in \R^{n\times l}$, as illustrated in Figures \ref{fig:gemm_naive} and \ref{fig:gemm_blocked}. We refer to these schemes as \emph{naive} and \emph{blocked} matrix multiplication, respectively. Detailed steps for these schemes are provided in Algorithms \ref{algo2} and \ref{algo3}, respectively, for the serverless setting. 
During naive matrix multiplication, each worker receives and multiplies an $a\times n$ row-block of $\A$ and $n\times a$ column-block of $\B$ to compute an $a\times a$ block of $\C$. 
Blocked matrix multiplication consists of two phases. During the computation phase, each worker gets two $b\times b$ blocks, one each from $\A$ and $\B$, which are then multiplied by the workers. In the reduction phase, to compute a $b \times b$ block of $\C$, one worker gathers results of all the $n/b$ workers from the cloud storage corresponding to one row-block of $\A$ and one column-block of $\B$ and adds them. For example, in Figure \ref{fig:gemm_blocked}, to get $\C(1,1)$, results from 3 workers who compute $\A(1,1)\times \B(1,1)$, $\A(1,2)\times \B(2,1)$ and $\A(1,3)\times \B(3,1)$ are added. 

It is accepted in High Performance Computing (HPC) that blocked partitioning of input matrices takes less time than naive matrix multiplication \cite{summa, 2.5d, demmel}. For example, in \cite{2.5d}, the authors propose 2.5D matrix multiplication, an optimal communication avoiding algorithm for matrix multiplication in HPC/server-based computing, that divides input matrices into blocks and stores redundant copies of them across processors to reduce bandwidth and latency costs. However, perhaps due to lack of a proper analysis for cloud-based distributed computing, existing algorithms for straggler mitigation in the cloud do naive matrix multiplication \cite{kangwook2,poly_codes,tavor}. Next, we bridge the gap between cost analysis and straggler mitigation for distributed computation in the serverless setting.

\begin{algorithm}[t]
\SetAlgoLined
\footnotesize
\SetKwInOut{Input}{Input}
\Input{Matrices $\A\in \R^{m\times n}$ and $\B\in \R^{n\times l}$}
\KwResult{$\C = \A\times \B$ }
 \textbf{Initialization}: Divide $\A$ into submatrices of size $a\times n$ (row-wise division) and $\B$ into submatrices of size $n \times a$ (column-wise division) \\
 \For{i=1 to m/a}{
  \For{j=1 to l/a}{
1. Worker $W_{ij}$ receives $i$-th chunk of $\A$, say $\A(i,:)$, and $j$-th chunk of $\B$, say $\B(:,j)$\\
2. $W_{ij}$ computes the $a\times a$ chunk of $\C$, that is, $\C(i,j) = \A(i,:) \times \B(:,j)$\\
3. $W_{ij}$ writes $\C(i,j)$ back to the cloud storage
 }}
 \caption{Distributed naive matrix multiplication}
 \label{algo2}
\end{algorithm}

\begin{algorithm}[t]
\SetAlgoLined
\footnotesize
\SetKwInOut{Input}{Input}
\Input{Matrices $\A\in \R^{m\times n}$ and $\B\in \R^{n\times l}$}
\KwResult{$\C = \A\times \B$ }
 \textbf{Initialization}: Divide $\A$ into $m/b\times n/b$ matrix and $\B$ into $n/b \times l/b$ matrix of $b\times b$ blocks where $b$ is the block-size\\
 \Comment*[h]{\color{blue}// Computation phase:}\\
 \For{i = 1 to m/b}{
  \For{j = 1 to l/b}{
  \For{k = 1 to n/b}{
1. Worker $W_{ijk}$ gets $(i,k)$-th block of $\A$, say $\A(i,k)$, and $(k,j)$-th block of $\B$, say $\B(k,j)$, from the cloud storage\\
2. $W_{ijk}$ then computes the $b\times b$ product $\hat{\C}_{ijk} = \A(i,k) \times \B(k,j)$\\
3. Worker writes the result $W_{ijk}$ back to the cloud storage
 }}}
\Comment*[h]{\color{blue}// Reduction phase:}\\
 \For{i = 1 to m/b}{
  \For{j = 1 to l/b}{
   Spin a new worker, say $W_{ij}$, that stores an all-zero $b \times b$ sub-block $C_{ij}$ \\
  \For{k = 1 to n/b}{
  1. $W_{ij}$ extracts the output $\hat{\C}_{ijk}$ written by $W_{ijk}$ from cloud storage\\
  2. $W_{ij}$ does $\C_{ij} = \C_{ij} + \hat{\C}_{ijk}$
 }
 $W_{ij}$ writes $\C_{ij}$ back to the cloud storage}}
 \caption{Distributed blocked matrix multiplication}
 \label{algo3}
\end{algorithm}

\section{Cost Analysis: Naive and Blocked multiplication}

There are communication and computation costs associated with any distributed algorithm. Communication costs themselves are of two types: latency and bandwidth. For example, sending $n$ bits requires packing them into contiguous memory and transmitting them as a message.
The latency cost $\alpha$ is the fixed overhead time spent in packing and transmitting a message over the network. Thus, to send $Q$ messages, the total latency cost is $\alpha Q$. 
Similarly, to transmit $K$ bits, a bandwidth cost proportional to $K$, given by $\beta K$, is associated.
Letting $\gamma$ denote the time to perform one floating point operation (FLOP),  the total computing cost is $\gamma F$, where $F$ is the total number of FLOPs at the node.  Hence, the total time pertaining to one node that sends $M$ messages, $K$ bits and performs $F$ FLOPs is 
$$T_{\text{worker}} = \alpha Q + \beta K + \gamma F,$$ 
where $\alpha \gg \beta \gg \gamma$. The $(\alpha,\beta, \gamma)$ model defined above has been well-studied and is used extensively in the HPC literature \cite{comm_cost2, demmel, demmel2, 2.5d, devarakonda}. It is ideally suited for serverless computing, where network topology does not affect the latency costs as each worker reads/writes directly from/to the cloud storage and no multicast gains are possible. 

However, our analysis for costs incurred during distributed matrix multiplication differs from previous works in three principle ways. 1) Workers in serverless architecture cannot communicate amongst themselves, and hence, our algorithm for blocked multiplication is very different from optimal communication avoiding algorithm for HPC that involves message passing between workers \cite{2.5d}. 2) The number of workers in HPC analyses is generally fixed, whereas the number of workers in the serverless setting is quite flexible, easily scaling into the thousands, and the limiting factor is memory/bandwidth available at each node. 3) Computation on the inexpensive cloud is more motivated by savings in expenditure than the time required to run the algorithm. 
We define our cost function below.

If there are $W$ workers, each doing an equal amount of work, the total amount of money spent in running the distributed algorithm on the cloud is proportional to 
\begin{equation}\label{Ctotal}
C_{\text{total}} = W\times T_{\text{worker}} = W(\alpha Q + \beta K + \gamma F).
\end{equation} 
Eq. \eqref{Ctotal} does not take into account the straggling costs as they increase the total cost by a constant factor (by re-running the jobs that are straggling) and hence does not affect our asymptotic analysis.

Inexpensive nodes in serverless computing are generally constrained by the amount of memory or communication bandwidth available. For example, AWS Lambda nodes have a maximum allocated memory of 3008 MB\footnote{AWS Lambda limits are available at (may change over time) https://docs.aws.amazon.com/lambda/latest/dg/limits.html},  a fraction of the memory available in today's smartphones. Let the memory available at each node be limited to $M$. That is, the communication bandwidth available at each worker is limited to $M$, and this is the main bottleneck of the distributed system. We would like to multiply two large matrices $\A\in \R^{m\times n}$ and $\B\in \R^{n\times l}$ in parallel, and let $M = O(n^\delta)$. 
Note that if $\delta \geq 2$, one of the following will happen: 
\begin{itemize}[leftmargin=*]
\item $m=O(n)$ and $l = O(n)$, and the input matrices can fit into one worker's memory and parallelism is not required.
\item Either $m=\omega(n)$ or $l=\omega(n)$ or both, and block-size for blocked matrix multiplication is $n$. The two schemes, naive and blocked multiplication, would exactly be the same in this case.
\end{itemize}
Thus, for all practical cases in consideration, $\delta < 2$.

\begin{theorem}\label{thm:cost}
For the cost model defined in Eq. \eqref{Ctotal}, communication (i.e., latency and bandwidth) costs for blocked multiplication outperform naive multiplication by a factor of $O(n^{1-\delta/2})$, where the individual costs are listed in Table \ref{table1}.
\end{theorem}
\begin{proof}
See Appendix \ref{appendix:cost_prop}.
\end{proof}

\begin{table}
  \caption{Costs comparison for naive and blocked matrix multiplication in the serverless setting, where $\delta <2$.}
  \label{table1}
  \centering
   \resizebox{0.9\columnwidth}{!}{
  \begin{tabular}{c|c|c|>{\columncolor[gray]{0.9}}c}
    \toprule
    \rowcolor{LightCyan}
     Cost type     & Naive multiply & Blocked Multiply & Ratio: naive/blocked \\
   \midrule
 Latency &  $O(mln^{2(1-\delta)})$ & $O(mln^{1-3\delta/2})$  & $O(n^{1-\delta/2})$     \\
 \midrule
Bandwidth &  $O(mln^{2-\delta})$ & $O(mln^{1-\delta/2})$  & $O(n^{1-\delta/2})$\\
\midrule
Computation &  $O(mln)$ & $O(mln)$  & $1$\\
    \bottomrule
  \end{tabular}
  }
\end{table}

The rightmost column in Table \ref{table1} lists the ratio of communication costs for naive and blocked matrix multiplication. We note that the latter significantly outperforms the former, with communication costs being asymptotically worse for naive multiplication.  An intuition behind why this happens is that each worker in distributed blocked multiplication does more work than in distributed naive multiplication for the same amount of received data. For example, to multiply two square matrices of dimension $n$, where memory at each worker limited by $M=2n$, $a=1$ for naive multiplication and $b = \sqrt{n}$ for blocked multiplication. We note that the amount of work done by each worker in naive and blocked multiplication is $O(n)$ and $O(n^{3/2})$, respectively. Since the total amount of work is constant and equal to $O(n^3)$, blocked matrix multiplication ends up communicating less during the overall execution of the algorithm as it requires fewer workers. 
Note that naive multiplication takes less time to complete as each worker does asymptotically less work, however, the number of workers required is asymptotically more, which is not an efficient utilization of resources and increases the expenditure significantly.

 \begin{remark}
For clarity of exposition, we partition the input matrices into square-blocks of dimension $b$. However, optimal block dimensions for rectangular-block partitions can be found by minimizing \eqref{Ctotal} while incorporating the memory constraints of the distributed system. This is a convex problem that can be converted to a geometric program \cite{gp_boyd}. We note that the optimal block dimension found in this way is nearly square.
 \end{remark}


\begin{figure}[t]
\centering
  \includegraphics[scale=0.55]{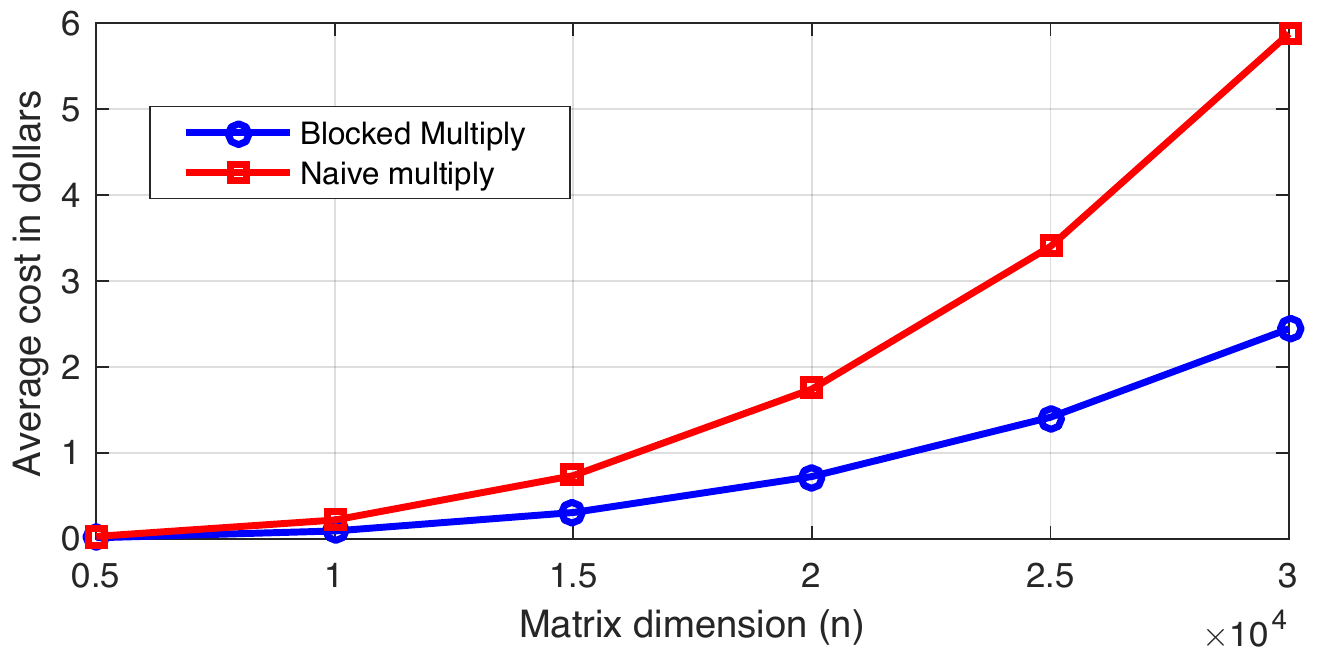}
  \caption{\small Comparison of AWS Lambda costs for multiplying two $n\times n$ matrices, where each worker is limited by 3008 MB of memory and price per running worker per 100 milliseconds is $\$0.000004897$.}
  \label{fig:lambda_costs}
  \vspace{-4mm}
\end{figure}

Figure \ref{fig:lambda_costs} supports the above analysis where we plot the cost in dollars of multiplying two square matrices in AWS Lambda, where each node's memory is limited by 3008 MB and price per worker per $100$ millisecond is \$$0.000004897$. However, as discussed earlier, existing schemes for straggler-resiliency in distributed matrix multiplication consider naive multiplication which is impractical from a user's point of view. 
In the next section, we propose OverSketch, a scheme to mitigate the detrimental effects of stragglers for blocked matrix multiplication.

\section{OverSketch: Straggler-resilient Blocked Matrix Multiplication using Sketching}

Many of the recent advances in algorithms for numerical linear algebra have come from the technique of linear sketching, in which a given matrix is compressed by multiplying it with a random matrix of appropriate dimension. The resulting product can then act as a proxy for the original matrix in expensive computations such as matrix multiplication, least-squares regression, low-rank approximation, etc. \cite{mahoney, woodruff_now, mert, yun_mert}. For example, computing the product of $\A\in \R^{m\times n}$ and $\B \in \R^{n\times l}$ takes $O(mnl)$ time. However, if we use $\s\in \R^{n\times d}$ to compute the sketches, say  $\tilde{\A} = \A\s \in \R^{m\times d}$ and $\tilde{\B} = \s^T\B \in \R^{d\times l}$, where $d \ll n$ is the sketch dimension, we can reduce the computation time to $O(mdl)$ by computing an approximate product $\A\s\s^T\B$. This is very useful in applications like machine learning, where the data itself is noisy, and computing the exact result is not needed. 

\textbf{Key idea behind OverSketch}:  Sketching accelerates computation by eliminating redundancy in the matrix structure through dimension reduction. However, the coding-based approaches described in Section \ref{related_work} have shown that redundancy can be  \emph{good} for combating stragglers if judiciously introduced into the computation. With these competing points of view in mind, our algorithm OverSketch works by "oversketching" the matrices to be multiplied by reducing dimensionality not to the minimum required for sketching accuracy, but rather to a slightly higher amount which simultaneously ensures both the accuracy guarantees and speedups of sketching \emph{and} the straggler resilience afforded by the redundancy which was not eliminated in the sketch.  OverSketch further reduces asymptotic costs by adopting the idea of block partitioning from HPC, suitably adapted for a serverless architecture. 

Next, we propose a sketching scheme for OverSketch and describe the process of straggler mitigation in detail.

\subsection{OverSketch: The Algorithm}

\begin{figure*}
\centering
\includegraphics[scale=0.47]{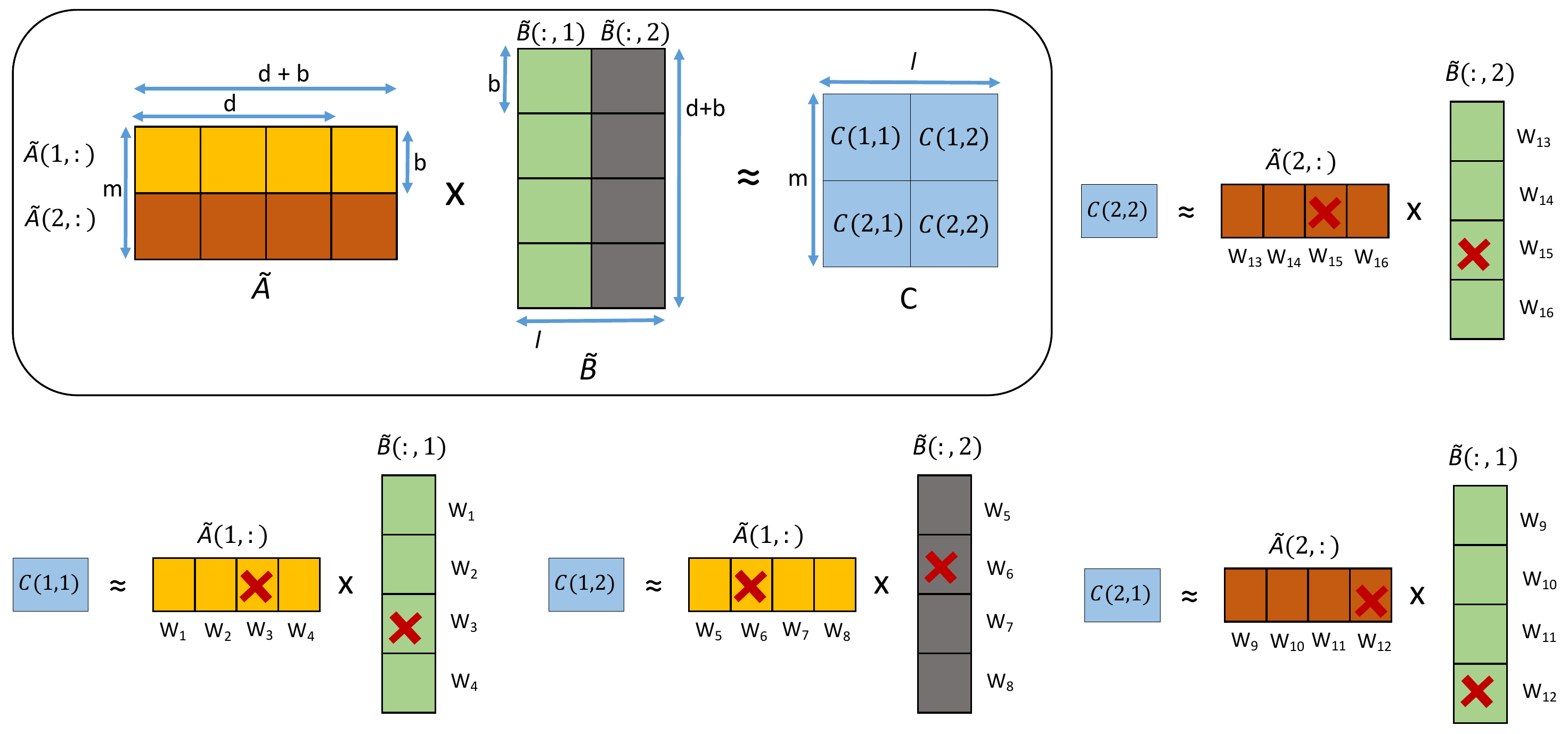}
\caption{\small An illustration of multiplication of $m\times z$ matrix $\tA$ and $z \times l$ matrix $\tB$, where $z = d+b$ assures resiliency against one straggler per block of $\C$, and $d$ is chosen by the user to guarantee a desired accuracy.  Here, $m = l = 2b$, $d = 3b$, where $b$ is the block-size for blocked matrix multiplication. This scheme ensures one worker can be ignored while calculating each block of $\C$.} 
\label{fig:bMM}
\end{figure*}

\begin{figure*}
\centering
\includegraphics[scale=0.55]{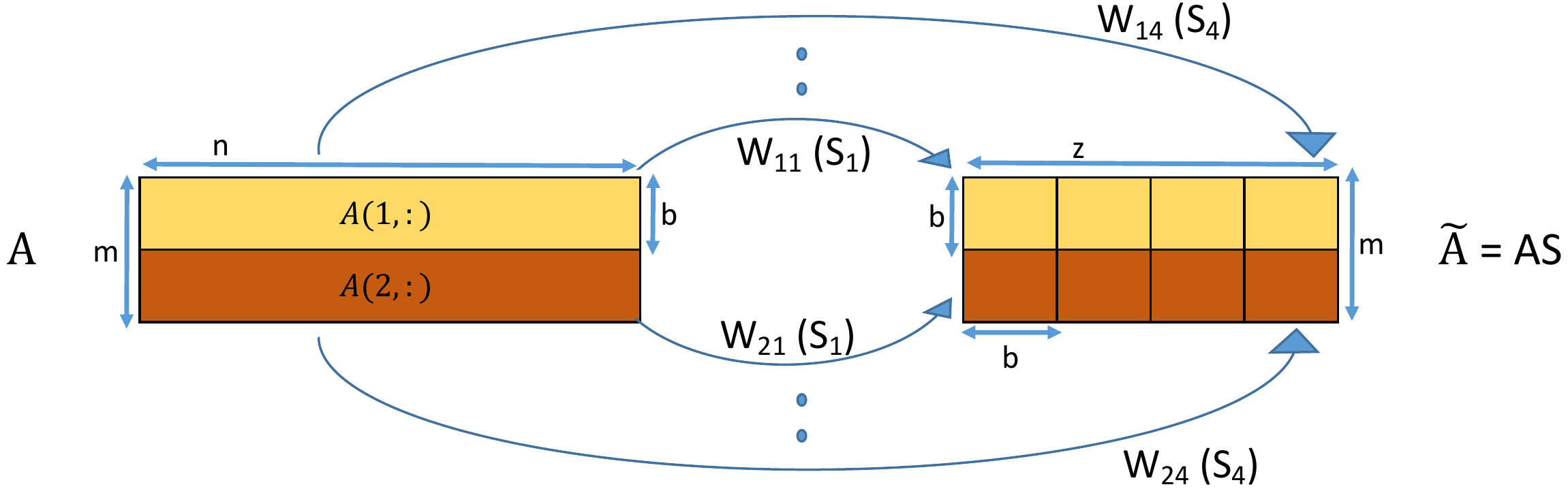}
\caption{\small An illustration of sketching $\A \in \R^{m\times n}$ in parallel using the sketch matrix in Eq. \eqref{sketch_matrix} with sketch dimension $z = (N+e)b$.  Worker $W_{ij}$ receives the row-block $\A(i,:)$ of $\A$ and the Count-sketch $\s_j$ to compute the $(i,j)$-th block of $\tA$. Sketching requires a total of $mz/b^2$ workers. Here, $z = 4b, N=3$ and $e=1$, and $\A$ is divided into 2 row-blocks, that is, $m = 2b$. Total number of workers required for distributed sketching is 8.}
\label{fig:oversketching}
\vspace{-2mm}
\end{figure*}

During blocked matrix multiplication, the $(i,j)$-th block of $\C$ is computed by assimilating results from $d/b$ workers who compute the product $\tA(i,k)\times \tB(k,j)$, for $k = 1,\cdots,d/b$. Thus, the computation $\C(i,j)$ can be viewed as the product of the row sub-block $\tilde{\A}(i,:)\in \R^{b\times d}$ of $\tilde{\A}$ and the column sub-block $\tilde{\B}(:,j) \in \R^{d\times b}$ of $\tilde{\B}$. An illustration is shown in Figure \ref{fig:bMM}. 
Assuming $d$ is large enough to guarantee the required accuracy in $\C$, we increase the sketch dimension from $d$ to $z = d + eb$, where $e$ is the worst case number of stragglers in $N = d/b$ workers. For the example in Figure \ref{fig:bMM}, $e=1$. To get a better insight on $e$, we observe in our simulations for cloud systems like AWS lambda and EC2 that the number of stragglers is $< 5\%$ for most runs. Thus, if $N = d/b = 40$, i.e. $40$ workers compute one block of $\C$, then $e \approx 2$ is sufficient to get similar accuracy for matrix multiplication. We describe OverSketch in detail in Algorithm \ref{algo:OverSketch}. Next, we describe how to compute the sketched matrices $\tA$ and $\tB$.

\begin{algorithm}[t]
\SetAlgoLined
\small
\SetKwInOut{Input}{Input}
\Input{Matrices $\A\in \R^{m\times n}$ and $\B\in \R^{n\times l}$, sketch dimension $z$, straggler tolerance $e$}
\KwResult{$\C \approx \A\times \B$ }
\textbf{Sketching}: Use Algorithm \ref{algo:oversketching} to obtain $\tA = \A\s$ and $\tB = \s^T\B$ distributedly\\
\textbf{Block partitioning}: Divide $\tA$ into $m/b\times z/b$ matrix and $\B$ into $z/b \times l/b$ matrix of $b\times b$ blocks where $b$ is the block-size\\
\textbf{Computation phase}: Use the computation phase from Algorithm \ref{algo3} to multiply $\tA$ and $\tB$. This step invokes $mlz/b^3$ workers, where $z/b$ workers are used per block of $\C$\\
\textbf{Termination}: Stop computation when any $d/b$ workers return their results for each of the $ml/b^2$ blocks of $\C$, where $d = z - eb$\\
\textbf{Reduction phase}: Invoke $ml/b^2$ workers for reduction as described in Algorithm \ref{algo3} on available results
 \caption{OverSketch: Distributed blocked matrix multiplication for the Cloud}
 \label{algo:OverSketch}
\end{algorithm}

Many sketching techniques have been proposed recently for approximate matrix computations. For example, to sketch a $m\times n$ matrix $\A$ with sketch dimension $d$, Gaussian projection takes $O(mnd)$ time, Subsampled Randomized Hadamard Transform (SRHT) takes $O(mn\log n)$ time, and Count-sketch takes $O(nnz(A))$ time, where $nnz(\cdot)$ is the number of non-zero entries \cite{clarkson_woodruff_stoc, woodruff_now, shusen, meng_mahoney}. 

\begin{algorithm}[t]
\SetAlgoLined
\small
\SetKwInOut{Input}{Input}
\Input{Matrix $\A\in \R^{m\times n}$ and sketch dimension $b$}
\KwResult{$\tA\in \R^{m\times b}$, a random Count-sketch of $\A$}
Multiply each column of $\A$ by $-1$ with probability 0.5\\
Map each column of resultant $\A$ to an integer in $[1,b]$ uniformly randomly\\
Add columns in $\A$ that are mapped to the same integer. The resultant matrix is a Count-sketch of $\A$ 
 \caption{Calculating Count-sketch of a matrix $\A$}
 \label{algo:count_sketch}
\end{algorithm}

Count sketch, first exploited in data streaming literature, has been widely applied to expedite large-scale matrix computations \cite{clarkson_woodruff_stoc, meng_mahoney, pagh}. It is one of the most popular sketching techniques as it requires linear time to compute the matrix sketch with similar approximation guarantees for matrix multiplication. 
To compute the Count-sketch of $\A\in \R^{m\times n}$ of sketch dimension $b$, each column in $\A$ is multiplied by $-1$ with probability 0.5 and then mapped to an integer sampled uniformly from $\{1,2,\cdots,b\}$. Then, to compute the sketch $\tA_c = \A\s_c$, columns with the same mapped value are summed (see Algorithm \ref{algo:count_sketch} for details). 
An example of Count-sketch matrix with $n=9$ and $b=3$ is
\begin{equation}\label{count_sketch_example}
\s_c^T = 
\begin{bmatrix}
    0  & 0 & 0 & 1 & -1 & 0 & -1  & 0 & 0 \\
    1  & -1 & 0 & 0 & 0 & 1 & 0  & 0 & 0\\
    0  & 0 & 1 & 0 & 0 & 0 & 0  & -1& -1
    \end{bmatrix}.
\end{equation}
Here, $\A$ has 9 columns, and columns $4,5$ and $7$ were mapped to $1$, columns $1,2$ and $6$ were mapped to 2, and columns $3,8$ and $9$ were mapped to $3$. Thus, the Count-sketch $\tA_c$ would have only 3 columns, which are obtained by summing the columns of $\A$ with the same mapped value (after possibly multiplying with -1). 
The sparse structure of $\s_c$ ensures that the computation of sketch takes $O(nnz(A))$ time.
However, a drawback of the desirable sparse structure of Count-sketch is that it cannot be directly employed for straggler mitigation in blocked matrix multiplication as it would imply complete loss of information from a subset of columns of $\A$. For the example in \eqref{count_sketch_example}, suppose the worker processing column 3 of $\tA_c$ be straggling. Ignoring this worker would imply that columns $3, 8$ and $9$ of $\A$ were not considered in the computation. This will generally lead to poor accuracy for sketched matrix multiplication.

To facilitate straggler mitigation for blocked matrix multiplication, we propose a new sketch matrix $\s$, inspired by Count-sketch, and define it as  
\begin{equation}\label{sketch_matrix}
\s = \frac{1}{\sqrt{N}}(\s_1, \s_2,\cdots,\s_{N+e}),
\end{equation}
where $N = d/b$, $e$ is the expected number of stragglers per block of $\C$ and $\s_i \in\R^{n\times b}$, for $i=1,2,\cdots,(N+e)$, is a Count-sketch matrix  with dimension $b$. Thus, the total sketch-dimension for the sketch matrix in \eqref{sketch_matrix} is $z = (N+e)b = d + eb$.
Computation of this sketch takes $O(nnz(\A)(N+e))$ time in total and can be implemented in a distributed fashion trivially, where $(N+e)$ is the number of workers per block of $\C$. An illustration of distributed sketching of $\A$ in serverless systems is described in Figure \ref{fig:oversketching}. For detailed steps, see Algorithm \ref{algo:oversketching}. A few remarks regarding OverSketch based distributed matrix multiplication are in order.

\begin{itemize}[leftmargin=*]
\setlength{\itemindent}{0em}
\item 
\textbf{Graceful degradation}: Coding-based straggler mitigation (see Figure \ref{fig:kangwook2} for example) cannot tolerate more stragglers than provisioned. An advantage of using sketching schemes for computation is that more stragglers can be tolerated than initially provisioned at the cost of accuracy of the result, thus exhibiting `graceful degradation'.

\item 
\textbf{Memory constrained distributed sketching}: It is assumed in Algorithm \ref{algo:oversketching} that each worker can store an entire row-block of $\A$ in memory to calculate its Count-sketch. That might not always be the case, especially in low-memory serverless nodes. However, since Count-sketch is a streaming based algorithm, workers can calculate the sketch by further partitioning the $b\times n$ row-block into $b\times b$ square blocks and copying only one block at a time (see \cite{count_stream} for details).

\item
\textbf{Straggler-resilient sketching}: We note that the stragglers can also be ignored during distributed sketching (Algorithm \ref{algo:oversketching}). More specifically, the blocks ignored during the sketching phase can be marked as faults/stragglers during computation phase in Algorithm \ref{algo:OverSketch}.

\item 
\textbf{Limitations of ``Over-sampling''}: Schemes like leverage-score based sampling are also used in literature to compute approximate product of $\A$ and $\B$ \cite{mahoney, mahoney2}. Such schemes are as efficient as Count-sketch but are not suitable for straggler-resilient blocked multiplication. For example, if a worker with a block of $\tA$ straggles, where $\tA$ is obtained by ``over''-sampling the columns of $\A$ according to leverage scores, results from all other workers that are working on that column-block of $\tA$ is wasted as part of the column-block is unavailable. Thus, oversampling can require huge redundancy even for a small number of stragglers.

\end{itemize}

Next, we prove statistical guarantees on the accuracy of our sketching based matrix multiplication algorithm.

\begin{algorithm}[t]
\SetAlgoLined
\small
\SetKwInOut{Input}{Input}
\Input{Matrix $\A\in \R^{m\times n}$ and ``Over'' sketch dimension $z$}
\KwResult{$\tA = \A\times \s$ }
 \textbf{Initialization}: Divide $\A$ into row-blocks of size $b\times n$ \\
 \For{i=1 to m/b}{
  \For{j=1 to z/b}{
1. Worker $W_{ij}$ receives $i$-th row-block of $\A$, say $\A(i,:)$\\
2. $W_{ij}$ uses Algorithm \ref{algo:count_sketch} to compute the $(i,j)$-th  $b\times b$ block of $\tA$ using Count-sketch $\s_j$, that is, $\tA(i,j) = \A(i,:)\times \s_j$\\
3. $W_{ij}$ writes $\tA(i,j)$ back to the cloud storage
 }}
 \caption{``Over'' sketching $\A$ in parallel using the sketch $\s$ in \eqref{sketch_matrix} to compute $\tA = \A\s$}
 \label{algo:oversketching}
\end{algorithm}

\subsection{OverSketch: Approximation guarantees}

\begin{definition}
We say that an approximate matrix multiplication of two matrices $\A$ and $\B$ using sketch $\s$, given by $\A\s\s^T\B$,  is $(\epsilon,\theta)$ accurate if, with probability at least $(1-\theta)$, it satisfies
\begin{equation*}
||\A\B - \A\s\s^T\B||_F^2 \leq \epsilon||\A||_F^2||\B||_F^2.
\end{equation*}
\end{definition}
Now, for blocked matrix multiplication using OverSketch and as illustrated in Figure \ref{fig:bMM}, the following holds
\begin{theorem}\label{main_thm}
Computing $(\A\s)\times(\s^T\B)$ using sketch $\s\in \R^{n\times z}$ in \eqref{sketch_matrix} and $d = \frac{2}{\epsilon\theta}$, while ignoring $e$ stragglers among any $\frac{z}{b}$ workers, is $(\epsilon,\theta)$ accurate.
\end{theorem}
\begin{proof}
See Appendix \ref{appendix:proof}.

For certain applications, the guarantee in theorem \ref{main_thm} may be too crude as the product of $||\A||_F^2$ and $||\B||_F^2$ in the RHS can get big for large matrices $\A$ and $\B$. We can obtain a stronger result than in theorem \ref{main_thm} when $\min(rank(\A), rank(\B)) \ll n$, for example, when $\A$ is a fat matrix, or $\B$ is a tall matrix. Without loss of generality, say $\min(rank(\A), rank(\B))  = rank(\A) = r$. Thus, $||\A||_2 \leq ||\A||_F\leq \sqrt{r}||\A||_2$, where $||\cdot||_2$ denotes the spectral norm.
Hence, with probability at least $(1-\theta)$
\begin{equation*}
||\A\s\s^T\B - \A\B||_F^2 \leq \epsilon r||\A||_2^2||\B||_F^2.
\end{equation*}
Now, if we increase the sketch dimension by a factor of $r$ to $z = r(d+eb) = O(\frac{r}{\epsilon\theta})$, we get
\begin{equation}
||\A\s\s^T\B - \A\B||_F^2 \leq \epsilon ||\A||_2^2||\B||_F^2
\end{equation}
with probability $(1-\theta)$, which is a better approximation for the product $\A\s\s^T\B$.

During the reduction phase, we use $ml/b^2$ workers, which is much less than the number of workers used during the computation phase, that is, $mlz/b^3$. In our experiments, we observe that the possibility of stragglers reduces significantly if fewer workers are used. This is especially true for the reduction phase, as healthy running workers from the computation phase are reused, reducing the chances of stragglers. 
However, in the unfortunate event that stragglers are observed during reduction, speculative execution can be used, i.e. detecting and restarting the slow job. Another simple solution is to use existing coding techniques as described in Figure \ref{fig:kangwook2}, that is, by adding one parity row-block to $\tA$ and one parity row column to $\tB$ before multiplying them, which can tolerate $3$ stragglers in the worst case. However, this would require a decoding step to compensate for the missing stragglers.

\section{Experimental Results} 
\subsection{Blocked Matrix Multiplication on AWS Lambda}

 \begin{figure*}[t!]
    \centering
    \begin{subfigure}[t]{0.48\textwidth}
        \centering
        \includegraphics[scale=0.75]{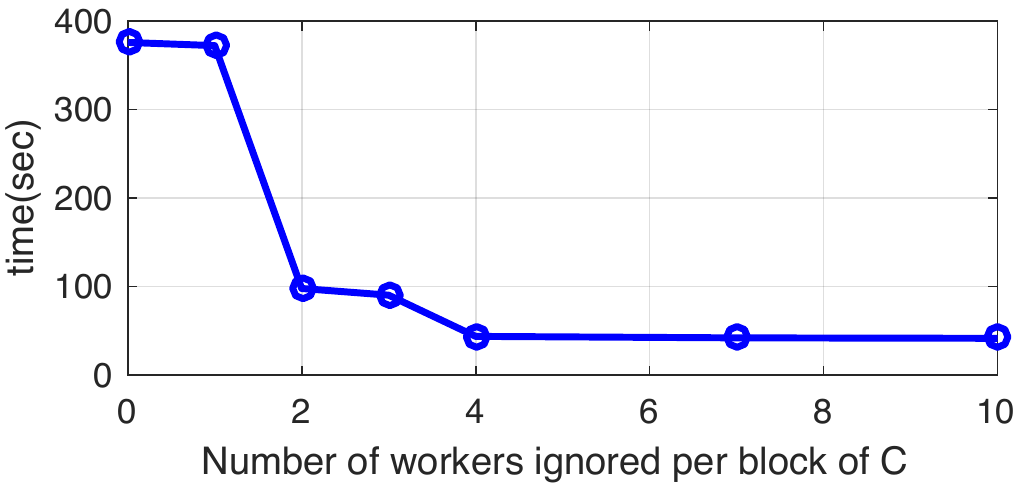}
        \caption{Time statistics for OverSketch on AWS Lambda for the straggler profile in Figure \ref{fig:stragglers}}
        \label{fig:lambda_matmul}
    \end{subfigure}
    ~
    \begin{subfigure}[t]{0.48\textwidth}
        \centering
        \includegraphics[scale=0.705]{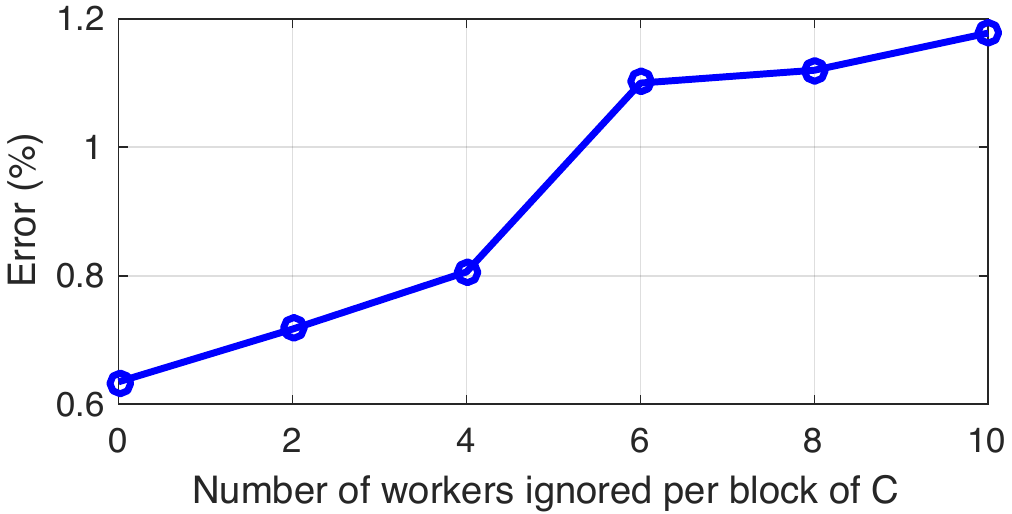}
        \caption{Frobenius norm error for sketched matrix product}
        \label{fig:sketch_error}
    \end{subfigure}
    \caption{\small Time and approximation error for OverSketch with 3000 workers when $e$, the number of workers ignored per block of $\C$, is varied from $0$ to $10$.}
\end{figure*}

\begin{figure*}[t]
    \centering
    \begin{subfigure}[t]{0.32\textwidth}
        \centering
        \includegraphics[scale=0.55]{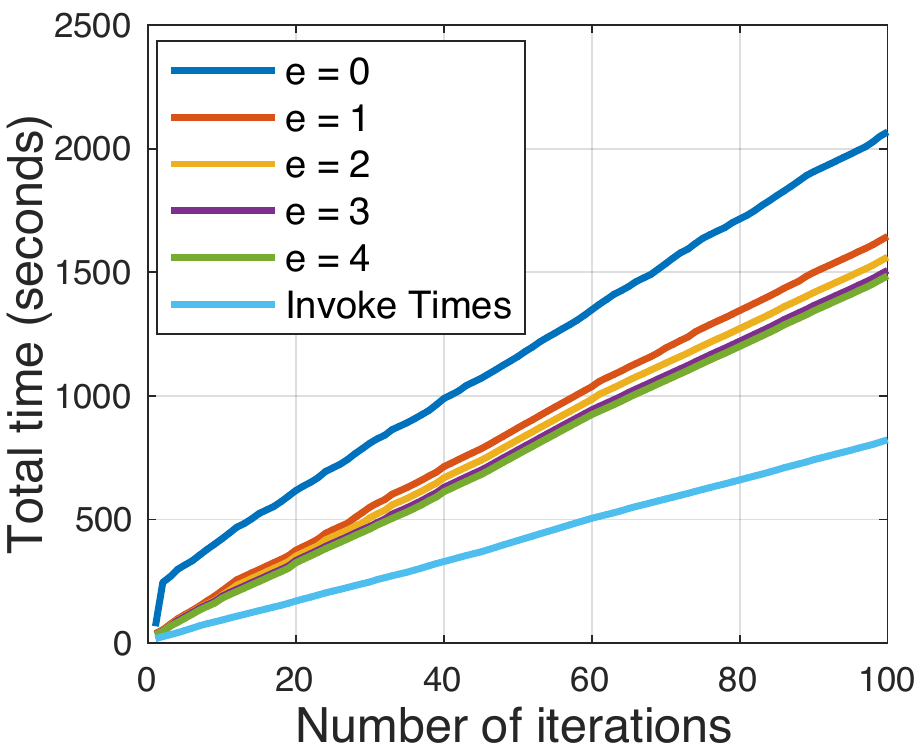}
        \caption{Plot of total time (i.e., invocation time plus computation time) versus number of iterations.}
        \label{fig:LP_with_invoke}
    \end{subfigure}
    ~
  \begin{subfigure}[t]{0.32\textwidth}
        \centering
        \includegraphics[scale=0.55]{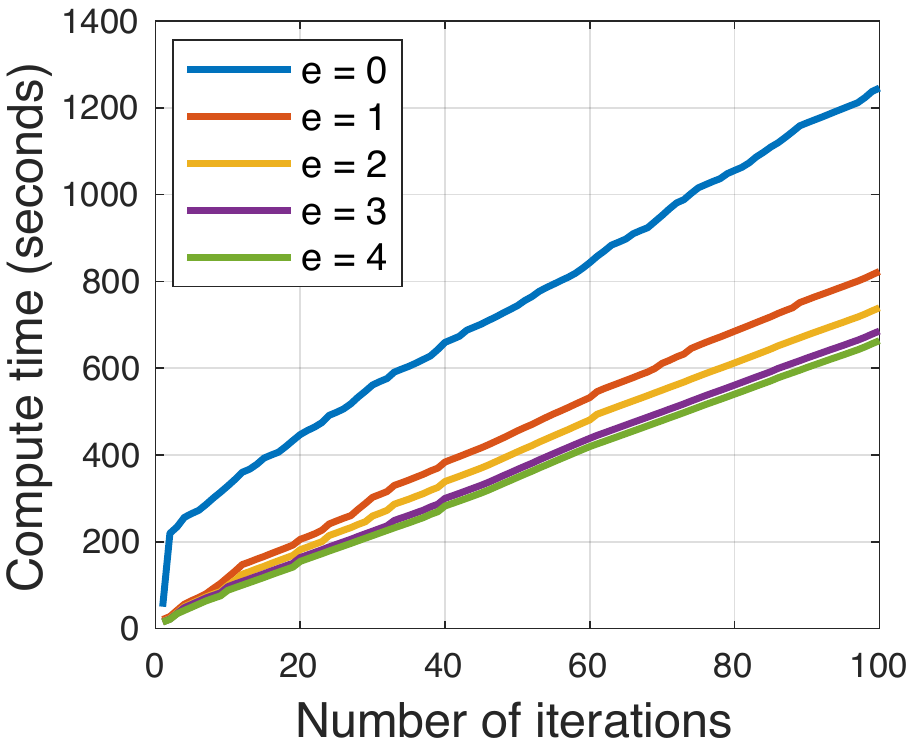}
        \caption{Plot of computation time versus number of iterations. When $e=1$, algorithm takes 7 minutes less compared to when $e=0$.}
        \label{fig:LP_without_invoke}
    \end{subfigure}%
    ~
    \begin{subfigure}[t]{0.32\textwidth}
        \centering
        \includegraphics[scale=0.55]{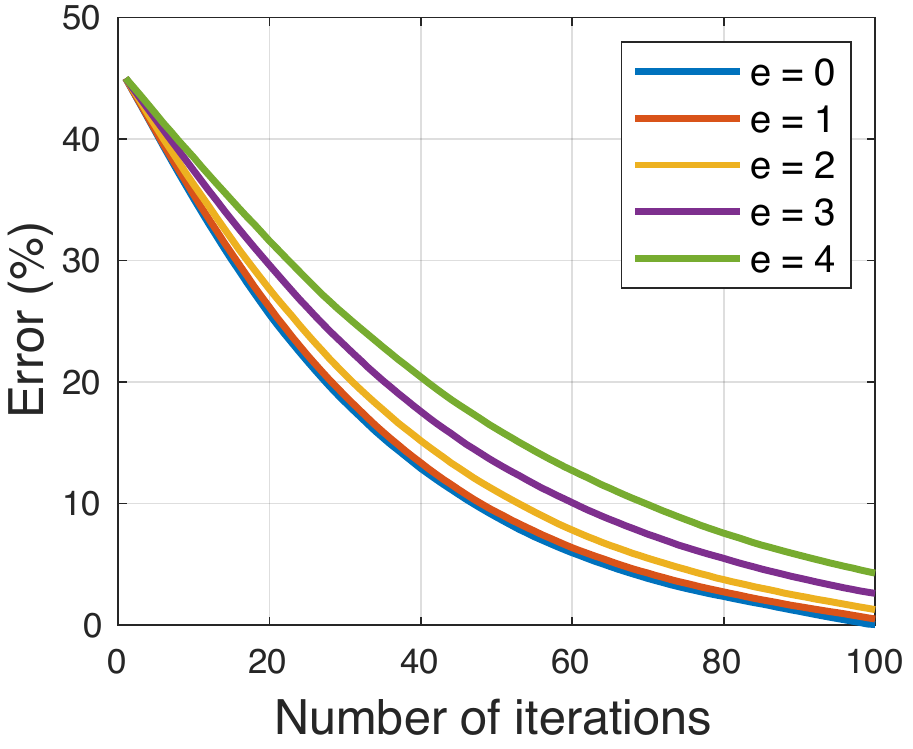}
        \caption{Plot of percentage error versus iterations. Ignoring one worker per block of $\C$ has negligible affect on the convergence.}
        \label{fig:LP_error}
    \end{subfigure}
    \caption{\small Time statistics and optimality gap on AWS Lambda while solving the LP in \eqref{lp} using interior point methods, where $e$ is the number of workers ignored per block of $\C$.}
    \label{fig:LPplots}
\end{figure*} 

We implement the straggler-resilient blocked matrix multiplication described above in the serverless computing platform \emph{Pywren} \cite{pywren, numpywren}\footnote{A working implementation of OverSketch is available at https://github.com/vvipgupta/OverSketch}, on the AWS Lambda cloud system to compute an approximate $\C = \A\s\times \s^T\B$ with $b = 2048, m=l =10b, n = 60b$ and $\s$ as defined in \eqref{sketch_matrix} with sketch dimension $z = 30b$. Throughout this experiment, we take $\A$ and $\B$ to be constant matrices where the entries of $\A$ are given by $\A(x,y) = x+y$ for all $x\in [1, m]$ and  $y\in [1,n]$ and $\B = \A^T$.  Thus, to compute $(i,j)$-th $b\times b$ block of $\C$, 30 nodes compute the product of $\tA(i,:)$ and $\tB(:,j)$, where $\tA=\A\s$ and $\tB=\s^T\B$. While collecting results, we ignore $e$ workers for each block of $\C$, where $e$ is varied from $0$ to $10$. 

The time statistics are plotted in Figure \ref{fig:lambda_matmul}. The corresponding worker job times are shown in Figure \ref{fig:stragglers}, where the median job time is around $42$ seconds, and some stragglers return their results around $100$ seconds and some others take up to $375$ seconds. We note that the compute time for matrix multiplication reduces by a factor of $9$ if we ignore at most $4$ workers per $30$ workers that compute a block of $\C$. In figure \ref{fig:sketch_error}, for same $\A$ and $\B$, we plot average error in matrix multiplication by generating ten instances of sketches and averaging the error in Frobenius norm, $\frac{||\A\B - \A\s\s^T\B||_F}{||\A\B||_F}$, across instances. We see that the average error is only $0.8\%$ when 4 workers are ignored. 


\subsection{Solving Optimization Problems with Sketched Matrix multiplication}

Matrix multiplication is the bottleneck of many optimization problems. Thus, sketching has been applied to solve several fairly common optimization problems using second-order methods, like linear programs, maximum likelihood estimation, generalized linear models like least squares and logistic regression, semi-definite programs, support vector machines, Kernel ridge regression, etc., with essentially same convergence guarantees as exact matrix multiplication \cite{mert, yun_mert}. As an instance, we solve the following linear program (LP) using interior point methods on AWS Lambda
\begin{align}\label{lp}
\underset{\x}{\text{minimize}}~ &\cc^T\x\\
\text{subject to}~ &\A\x\leq \bb,\nn
\end{align}
where $\x \in \R^{m\times 1}, \cc \in \R^{m\times 1}, \bb \in \R^{n\times 1}$ and $\A\in \R^{n\times m}$ is the constraint matrix with $n>m$. To solve \eqref{lp} using the logarithmic barrier method, we solve the following sequence of problems using Newton's method
\begin{equation}\label{int_point}
\min_{\x\in \R^m}  f(\x) = \min_{\x\in \R^m} \left(\tau \cc^T\x - \sum_{i=1}^n\log(b_i - \ba_i\x)\right),
\end{equation}
where $\ba_i$ is the $i$-th row of $\A$, $\tau$ is increased geometrically as $\tau = 2\tau$ after every 10 iterations and the total number of iterations is $100$. The update in the $t$-th iteration is given by 
\begin{equation}
\x_{t+1} = \x_t - \eta(\nabla^2f(\x_t))^{-1}\nabla f(\x_t),
\end{equation}
where $\x_t$ is the estimate of the solution in the $t$-th iteration and $\eta$ is the appropriate step-size. The gradient and Hessian for the objective in \eqref{int_point} are given by
\begin{eqnarray}
\nabla f(\x) = \tau \cc + \sum_{i=1}^n \frac{\ba_i^T}{b_i - \ba_i^T\x} \text{ and} \\
\nabla^2f(\x) = \A^T\text{diag}\frac{1}{(bi - \ba_i\x)^2}\A,
\end{eqnarray}  
respectively. The square root of the Hessian is given by $\nabla^2f(\x)^{1/2} = \text{diag}\frac{1}{|bi - \ba_i\x|}\A$. The  computation of Hessian requires $O(nm^2)$ time and is the bottleneck in each iteration. Thus, we use our distributed and sketching-based blocked matrix multiplication scheme to mitigate stragglers while evaluating the Hessian approximately, i.e. $\nabla^2f(\x) \approx (\s\nabla^2f(\x)^{1/2})^T\times (\s\nabla^2f(\x)^{1/2})$, on AWS Lambda, where $\s$ is defined in \eqref{sketch_matrix}. 

We take the block size, $b$, to be $1000$, the dimensions of $\A$ to be $n=40b$ and $m = 5b$ and the sketch dimension to be $z=20b$. We use a total of $500$ workers in each iteration. Thus, to compute each $b \times b$ block of $\C$, $20$ workers are assigned to compute matrix multiplication on two $b\times b$ blocks. We depict the time and error versus iterations in figure \ref{fig:LPplots}. We plot our results for different values of $e$, where $e$ is the number of workers ignored per block of $\C$. In our simulations, each iteration includes around $9$ seconds of invocation time to launch AWS Lambda workers and assign tasks. In figure \ref{fig:LP_with_invoke}, we plot the total time that includes the invocation time and computation time versus iterations. In \ref{fig:LP_without_invoke}, we exclude the invocation time and plot just the compute time in each iteration and observe $34\%$ savings in solving \eqref{lp} when $e=1$, whereas the effect on the error with respect to the optimal solution is insignificant (as shown in figure \ref{fig:LP_error}).  

\subsection{Comparison with Existing Straggler Mitigation Schemes}

\begin{figure}[t!]
    \centering
    \begin{subfigure}[t]{0.23\textwidth}
        \centering
        \includegraphics[scale=0.42]{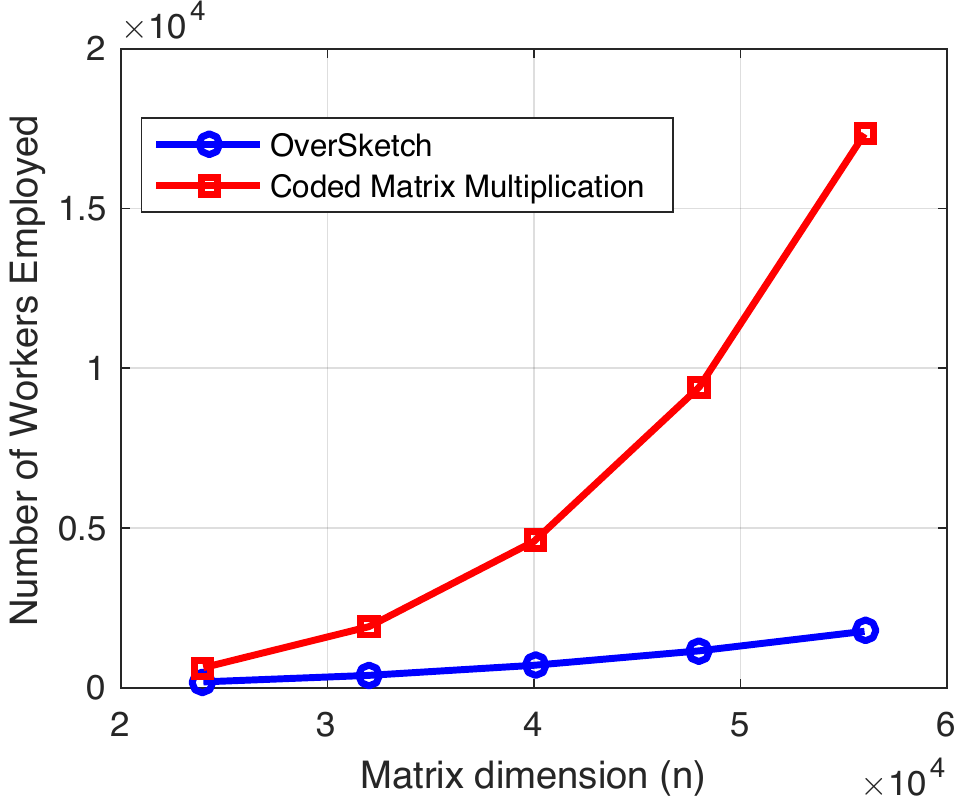}
        \caption{Number of workers required for OverSketch and \cite{kangwook2}}
        \label{fig:worker_comparison}
    \end{subfigure} 
    ~
    \begin{subfigure}[t]{0.23\textwidth}
        \centering
        \includegraphics[scale=0.42]{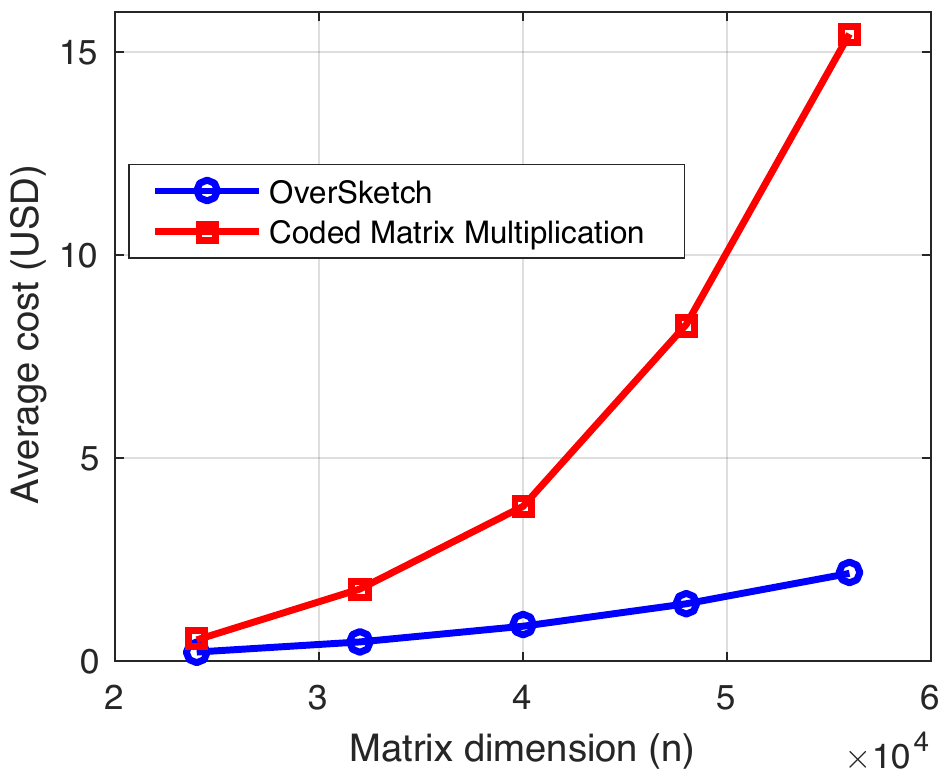}
        \caption{\footnotesize Cost for distributed matrix multiplication for OverSketch and \cite{kangwook2}}
        \label{fig:cost_comparison}
    \end{subfigure}
    \caption{\small Comparison of OverSketch with coded theory based scheme in \cite{kangwook2} on AWS Lambda. OverSketch requires asymptotically less workers which translates to significant savings in cost.}
    \vspace{-4mm}
\end{figure} 

In this section, we compare OverSketch with an existing coding-theory based straggler mitigation scheme described in \cite{kangwook2}. An illustration for \cite{kangwook2} is shown in Figure \ref{fig:kangwook2}. We multiply two square matrices $\A$ and $\B$ of dimension $n$ on AWS Lambda using the two schemes, where $\A(x,y) = x+y$ and $\B(x,y) = x\times y$ for all $x,y \in [1,n]$.  We limit the bandwidth of each worker by 400 MB (i.e. around 48 million entries, where each entry takes 8 bytes) for a fair comparison. Thus, we have $3b^2 = 48\times 10^6$, or $b=4000$ for OverSketch and $2an + a^2 = 48\times 10^6$ for \cite{kangwook2}, where $a$ is the size of the row-block of $\A$ (and column-block of $\B$). We vary the matrix dimension $n$ from $6b = 24000$ to $14b = 56000$. For OverSketch, we take the sketch dimension $z$ to be $n/2 + b$, and take $e=1$, i.e., ignore one straggler per block of $\C$. For straggler mitigation in \cite{kangwook2}, we add one parity row in $\A$ and one parity column in $\B$. In Figures \ref{fig:worker_comparison} and \ref{fig:cost_comparison}, we compare the workers required and average cost in dollars, respectively, for the two schemes. We note that OverSketch requires asymptotically fewer workers, and it translates to the cost for doing matrix multiplication. This is because the running time at each worker is heavily dependent on communication, which is the same for both the schemes. For $n=20000$, the average error in Frobenius norm for OverSketch is less than $2\%$, and decreases as $n$ is increased. 

The scheme in \cite{kangwook2} requires an additional decoding phase, and assume the existence of a powerful master that can store the entire product $\C$ in memory and decode for the missing blocks using the redundant chunks. This is also true for the other schemes in \cite{poly_codes,tavor,matdot}. Moreover, these schemes would fail when the number of stragglers is more than the provisioned redundancy while OverSketch has a 'graceful degradation' as one can get away by ignoring more workers than provisioned at the cost of accuracy of the result.  

\section{Conclusion}
Serverless computing penetrates a large user base by allowing users to run distributed applications without the hassles of server management. 
We analyzed the cost of distributed computation in serverless computing for naive and blocked matrix multiplication. Through analysis and experiments on AWS Lambda, we show that the latter significantly outperforms the former. Thus, existing straggler mitigation schemes that do naive matrix multiplication are unsuitable. To this end, we develop OverSketch, a sketching based algorithm for approximate blocked matrix multiplication. Our sketching scheme requires time linear in the size of input matrices. 
As a distributed matrix multiplication algorithm, OverSketch has many advantages: reduction in dimension of input matrices for computational savings, and built-in straggler resiliency. 
Extensive experiments on AWS Lambda support our claims that OverSketch is resilient to stragglers, cost-efficient, and highly accurate for suitably chosen sketch dimension.

\appendix

\subsection{Proof of Theorem \ref{thm:cost}}\label{appendix:cost_prop}
To compare naive and blocked multiplication, we first observe that the computation cost in \eqref{Ctotal}, that is $W\times F$, is the same for both naive and blocked multiplication and is equal to $O(mnl)$, which is the total amount of work done during matrix-matrix multiplication\footnote{The computation cost for blocked matrix multiplication can be further improved by using Strassen type methods that take $O(b^{2.38})$ to multiply two square sub-blocks of dimension $b\times b$, but we do not consider that advantage in this paper for clarity of exposition and to emphasize on savings just due to communication.}. 
Let $W_1$ be the number of workers required for naive matrix multiplication. Then, $W_1 = ml/a^2,$ as each worker is sent one row-block of $\A$ from $m/a$ choices, and one column-block of $\B$ from $l/a$ choices. Each worker receives $2an$ entries and writes back $a^2$ entries. Hence, the total communication incurred during the algorithm is $W_1\times (2an + a^2) = (2nml/a + ml)$. Also, since each worker can only receive $M$ entries, we have $M = 2an$, thus $a = M/2n$. Hence, the total bandwidth cost for naive multiplication is $\beta \times (4n^2ml/M + ml)= O(n^{2-\delta}ml)$. Also, the total number of messages sent during the process is $W_1$, and hence the total latency cost is $O(mln^{2(1-\delta)})$.

During the computation phase for blocked multiplication, $W_{2,comp} =  (n/b)\times ml/b^2$, as computation of one $b\times b$ block of $\C\in \R^{m\times l}$ requires $n/b$ workers, and there are a total of $ml/b^2$ such blocks. Again, each worker receives two $b\times b$ blocks, one from each $\A$ and $\B$, and writes back a $b\times b$ block, where $b$ satisfies $M =2b^2$. Thus, the total bandwidth cost incurred during the computation phase is $\beta W_{2,comp}\times 3b^2 = 3\beta nml/b = O(mln^{1-\delta/2})$. The total number of messages received by the workers is $W_{2,comp}$, and, hence, the latency cost is $\alpha nml/b^3 = O(mln^{1-3\delta/2})$. During the reduction phase, the number of workers required is $W_{2,red} = ml/b^2$, and each worker receives $n/b$ blocks of size $b \times b$ to compute one block of $\C$. Thus, for the reduction phase, the communication is $W_{2,red} \times b^2 \times (n/b) = nml/b = O(mln^{1-\delta/2})$ and total messages sent is $W_{2,red}\times (n/b) = mln/b^3 = O(mln^{1-3\delta/2})$. Hence, the total latency and bandwidth costs for blocked multiplication are $O(mln^{1-3\delta/2})$ and $O(mln^{1-\delta/2})$, respectively. This analysis justifies the costs summarized in Table \ref{table1} and proves the theorem.

\subsection{Proof of Theorem \ref{main_thm}}\label{appendix:proof}

The following three lemmas will assist us with the proof of Theorem \ref{main_thm}. 
\begin{lemma}\label{countEV}
let $\s_c \in \R^{n\times b}$ be a Count sketch matrix. Then, for any vectors $\textbf{x},\textbf{y} \in \R^{n\times 1}$, the following holds
\begin{align}
&\e[\textbf{x}^T\s_c\s_c^T\textbf{y}] = \textbf{x}^T\textbf{y} \label{count_sketch_expectation}\\
&\textrm{Var}[\textbf{x}^T\s_c\s_c^T\textbf{y}] = \frac{1}{b}\left(\sum_{j\neq l} x_j^2y_l^2 + \sum_{j\neq l}x_jy_jx_ly_l\right)\nn \\ 
&\leq \frac{1}{b}\left((\textbf{x}^T\textbf{y})^2 + ||\textbf{x} ||_2^2||\textbf{y} ||_2^2\right) \leq \frac{2}{b}||\textbf{x} ||_2^2||\textbf{y} ||_2^2.\label{count_sketch_variance}
\end{align}
\end{lemma}
\begin{proof}
See \cite{varS},  Appendix A. 
\end{proof}

\begin{lemma}\label{lemma:sketch_properties}
Let $\s = \frac{1}{\sqrt{N}}(\s_1, \s_2,\cdots,\s_{N}) \in \R^{n\times d},$ where $d = Nb$ and $\s_i \in \R^{n\times b}$ is a Count-sketch matrix that satisfies \eqref{count_sketch_expectation} and \eqref{count_sketch_variance}, for all $i\in 1,2,\cdots,N$. Then, for any vectors $\textbf{x},\textbf{y} \in \R^{n\times 1}$, the following holds
\begin{eqnarray*}
\e[\textbf{x}^T\s\s^T\textbf{y}] &=& \textbf{x}^T\textbf{y}\\
\textrm{Var}[\textbf{x}^T\s\s^T\textbf{y}]  &\leq& \frac{2}{d}||\textbf{x} ||_2^2||\textbf{y} ||_2^2.
\end{eqnarray*}
\end{lemma}
\begin{proof}
Note that, $\s\s^T = \frac{1}{N}(\s_1\s_1^T + \s_2\s_2^T + \cdots + \s_N\s_N^T)$. Thus, 
$$\x^T\s\s^T\y = \frac{1}{N}(\x^T\s_1\s_1^T\y + \x^T\s_2\s_2^T\y + \cdots + \x^T\s_N\s_N^T\y),$$
and hence, $\e[\textbf{x}^T\s\s^T\textbf{y}] = \textbf{x}^T\textbf{y}$ by \eqref{count_sketch_expectation} and linearity of expectation. Now,
\begin{align}
&~~~~\textrm{Var}[\x^T\s\s^T\y] = \e[(\x^T\s\s^T\y - \x^T\y)^2]\nn\\
&=\e\scalebox{0.87}{$\left[\frac{1}{N}\left((\x^T\s_1\s_1^T\y + \x^T\s_2\s_2^T\y + \cdots + \x^T\s_N\s_N^T\y) - N\x^T\y\right)^2\right]$}\nn\\
&=\e\biggl[\frac{1}{N^2}\left(\sum_{i=1}^N(\x^T\s_i\s_i^T\y - \x^T\y)\right)^2\biggr]\nn\\
&=\frac{1}{N^2} \biggl(\sum_{i=1}^N\e[(\x^T\s_i\s_i^T\y - \x^T\y)^2] \nn\\ 
&~~~~~~~~+ \sum_{i\neq j}\e[(\x^T\s_i\s_i^T\y - \x^T\y)(\x^T\s_j\s_j^T\y - \x^T\y)]\biggr)\nn\\
&=\frac{1}{N^2}\biggl(\sum_{i=1}^N \textrm{Var}[\x^T\s_i\s_i^T\y] + \nn\\ 
&~~~~\sum_{i\neq j}\e[(\x^T\s_i\s_i^T\y - \x^T\y)(\x^T\s_j\s_j^T\y - \x^T\y)]\biggr). \label{var_sketch}
\end{align}
Noting that $\s_1,\s_2, \cdots, \s_N$ are independent random variables and using \eqref{count_sketch_expectation}, we get  
\begin{align*}
&\e[(\x^T\s_i\s_i^T\y - \x^T\y)(\x^T\s_j\s_j^T\y - \x^T\y)] \\ 
&= \e[\x^T\s_i\s_i^T\y - \x^T\y]\e[\x^T\s_j\s_j^T\y - \x^T\y] = 0 ~\forall~ i\neq j.
\end{align*}
Now, using the above equation and \eqref{count_sketch_variance} in \eqref{var_sketch}, we get 
$$\textrm{Var}[\x^T\s\s^T\y] = \frac{1}{N^2}\times N\times \frac{2}{b}||\textbf{x} ||_2^2||\textbf{y} ||_2^2 = \frac{2}{d}||\textbf{x} ||_2^2||\textbf{y} ||_2^2,$$
which proves the lemma.
\end{proof}

\begin{lemma}\label{lemma2}
Let $d = 2/\epsilon$. Then, for any $\A\in \R^{m\times n}$, $\B \in \R^{n\times l}$ and $\s$ as defined in lemma \ref{lemma:sketch_properties},
\begin{equation}
\e||\A\B - \A\s\s^T\B||_F^2 \leq \epsilon||\A||_F^2||\B||_F^2.
\end{equation}
\end{lemma}
\begin{proof}
By the property of Frobenius norm and linearity of expectation, we have 
{\small
\begin{eqnarray}\label{eq1}
\e||\A\B - \A\s\s^T\B||_F^2 = \sum_{i=1}^m\sum_{j=1}^l\e|\textbf{a}^{(i)}\textbf{b}_{(j)} - \textbf{a}^{(i)}\s\s^T\textbf{b}_{(j)}|^2,
\end{eqnarray}}
where $\textbf{a}^{(i)}$ and $\textbf{b}_{(j)}$ are the $i$-th row and $j$-th columns of $\A$ and $\B$, respectively. Now, using lemma \ref{lemma:sketch_properties} in \eqref{eq1}, we get
\begin{align*}
\e||\A\B &- \A\s\s^T\B||_F^2 = \sum_{i=1}^m\sum_{j=1}^k\textrm{Var}[\textbf{a}^{(i)}\s\s^T\textbf{b}_{(j)}]^2\\
 &\leq \sum_{i=1}^m\sum_{j=1}^k\frac{2}{d}||\textbf{a}^{(i)} ||_2^2||\textbf{b}_{(j)} ||_2^2\\
 &=  \epsilon\left(\sum_{i=1}^m||\textbf{a}^{(i)} ||_2^2\right)\biggl( \sum_{j=1}^k||\textbf{b}_{(j)} ||_2^2\biggr)  ~~~\text{ (as }d=2/\epsilon)\\
 &= \epsilon||\A||_F^2||\B||_F^2,
\end{align*}
which is the desired result.
\end{proof}

We are now ready to prove theorem \ref{main_thm}.
As illustrated in figure \ref{fig:bMM}, we can think of computation of a $b\times b$ sub-block $\C(i,j)$ as multiplication of row block $\tA(i,:)$ of $\tA = \A\s$ and column-block $\tB(:,j)$ of $\tB = \s^T\B$. Since we ignore upto only $e$ workers in the calculation of a $b\times b$ block of $\C$, the effective sketch dimension is greater than $d= \frac{2}{\epsilon\theta}$, and therefore, from lemma \ref{lemma2}
\begin{align}\label{eq2}
\e||\A(i,:)\B(:,j) &- \A(i,:)\s_{ij}\s_{ij}^T\B(:,j)||_F^2 \nn \\
&\leq \epsilon\theta||\A(i,:)||_F^2||\B(:,j)||_F^2, 
\end{align}
for all $ i\in 1,\cdots,m/b$ and $ j \in 1,\cdots, l/b$.
Note that even if we applied the same sketch on $\A$ and $\B$ across row and column blocks, respectively, $\s_{ij}$ in the above equation might end up being different for each pair $(i,j)$ depending upon the location of stragglers, though with a common property that the sketch dimension is at least $d$.  
Now, we note that
\begin{align*}
\e||&\A\s\s^T\B - \A\B||_F^2 \\
&= \sum_{i=1}^{m/b}\sum_{j=1}^{l/b} \e||\A(i,:)\B(:,j) - \A(i,:)\s_{ij}\s_{ij}^T\B(:,j)||_F^2\\
&\leq \epsilon\theta\sum_{i=1}^{m/b}\sum_{j=1}^{l/b}||\A(i,:)||_F^2||\B(:,j)||_F^2 = \epsilon\theta||\A||_F^2||\B||_F^2.
\end{align*}
Now, by Markov's inequality
\begin{align*}
\mathbb{P}(||\A\s\s^T\B &- \A\B||_F^2 > \epsilon||\A||_F^2||\B||_F^2) \\
&\leq \frac{\e||\A\s\s^T\B - \A\B||_F^2}{{\epsilon||\A||_F^2||\B||_F^2}} \\
&\leq \frac{\epsilon\theta||\A||_F^2||\B||_F^2}{\epsilon||\A||_F^2||\B||_F^2} = \theta,
\end{align*}
which proves the desired result.
\end{proof}

%
%



%
\footnotesize
\bibliographystyle{IEEEtran}
\bibliography{bibli}

\end{document}